\documentclass[10pt]{article}

\usepackage[latin9]{inputenc}
\usepackage{latexsym, bm, amsmath, amssymb, graphics, amsthm,
enumerate} %
\usepackage[a4paper,
            bindingoffset=0.2in,
            left=0.8in,
            right=0.8in,
            top=0.8in,
            bottom=0.8in,
            footskip=.25in]{geometry}
\usepackage{xcolor,placeins}
\usepackage[ruled,vlined]{algorithm2e}
\usepackage[round]{natbib}
\usepackage{tikz}

\usetikzlibrary{calc,trees,positioning,arrows,chains,shapes.geometric,  decorations.pathreplacing,decorations.pathmorphing,shapes,  matrix,shapes.symbols}
\usepackage{tikz-cd}
\tikzcdset{scale cd/.style={every label/.append style={scale=#1},
    cells={nodes={scale=#1}}}}

\usepackage{hyperref} 
\usepackage{multirow,multicol}

\usetikzlibrary{arrows.meta,
                positioning,
                shapes}
\newtheorem{theorem}{Theorem}

\newtheorem{lemma}{Lemma}[section]

\newtheorem{axiom}{Axiom}

\newtheorem{Ex}{Example}

\usepackage{comment}

\begin{document}

\title{Axiomatization of R\'enyi Entropy on Quantum Phase Space\footnote{We thank Samson Abramsky, Roberto Corrao, Mar\'ia Garc\'a D\'iaz, Pierre Tarres, Mark Wilde, Noson Yanofsky, Stuart Zoble, and conference audiences at NYU Shanghai and Southwestern University of Finance and Economics for important input.  Ruodu Wang generously shared a significant extension of the previous version of Theorem 3.  Financial support from NYU Stern School of Business, NYU Shanghai, J.P. Valles, and the HHL - Leipzig Graduate School of Management is gratefully acknowledged.  ChatGPT and Grok were used to help find proof tactics, create numerical examples, identify references, and prepare Figures 1 and 2.}}

\author
{Adam Brandenburger \footnote{Stern School of Business, Tandon School of Engineering, NYU Shanghai, New York University, New York, NY 10012, U.S.A., adam.brandenburger@stern.nyu.edu}
\and{Pierfrancesco La Mura \footnote{HHL - Leipzig Graduate School of Management, 04109 Leipzig, Germany, plamura@hhl.de}}
    }
\date{\today}
\maketitle

\thispagestyle{empty}

\begin{abstract}
Phase-space versions of quantum mechanics -- from Wigner's original distribution to modern discrete-qudit constructions -- represent some states with negative quasi-probabilities.  Conventional Shannon and R\'enyi entropies become complex-valued in this setting and lose their operational meaning.  Building on the axiomatic treatments of R\'enyi (1961) and Dar\'oczy (1963), we develop a conservative extension that applies to signed finite phase spaces and identify a single admissible entropy family, which we call signed R\'enyi $\alpha$-entropy (for a free parameter $\alpha \ge 0$).  The obvious signed Shannon candidate is ruled out because it violates extensivity.  We prove four results that bolster the usefulness of the new measure.  (i) It serves as a witness of the presence of cancellation, detecting the coexistence of positive and negative weight in a signed measure.  (ii) For $\alpha > 1$, it is Schur-concave, delivering the intuitive property that mixing increases entropy  (iii) The same parametric family obeys a quantum H-theorem, namely, that under de-phasing dynamics entropy cannot decrease.  (iv) The $2$-entropy is conserved under discrete Moyal-bracket dynamics, mirroring conservation of von Neumann entropy under unitary evolution on Hilbert space.  We also comment on interpreting the R\'enyi order parameter as an inverse temperature.  Overall, we believe that our investigation provides good evidence that our axiomatically derived signed R\'enyi entropy may be a useful addition to existing entropy measures employed in quantum information, foundations, and thermodynamics.
\end{abstract}

\section{Introduction}
Phase-space formulations of quantum mechanics -- beginning with Wigner's quasi-probability distribution (Wigner, 1932) and including Feynman's comments on negative probability (Feynman, 1987) and the discrete constructions for qudits (Wootters, 1987; Gibbons, Hoffman, and Wootters, 2004) -- assign negative weights to certain states.  When a logarithm operatives on these negative quantities, the familiar Shannon or R\'enyi entropies become complex-valued and lose their operational meaning in terms of the amount of uncertainty about a system. This poses an important foundational question:

\begin{quote}
What axioms must an entropy satisfy to remain real-valued and physically informative when its input can be negative?
\end{quote}

A clue to the answer can be found in the R\'enyi (1961) axiomatization of entropy that generalized Shannon entropy (Shannon, 1948) and which R\'enyi later used in various problems in probability theory and information theory.  (Applications to ergodicity of Markov chains and the central limit theorem can be found in R\'enyi, 1970, pp.598-603.)  Since then, R\'enyi entropy has been used in many fields (Csisz\'ar, 2008).  Two instances are quantum information and quantum thermodynamics, which are the primary motivation for the current paper.  See Coles et al.~(2017) and Goold et al.~(2016) for comprehensive surveys.

The usual definition of entropy in quantum mechanics is von Neumann entropy (von Neumann, 1932), which is the natural analog to Shannon entropy.  One can also define the Hilbert-space analog to R\'enyi entropy, as in the works above.  Our goal is to find a natural definition on phase space.

In this paper, we return to the R\'enyi-Dar\'oczy axioms (R\'enyi, 1961; Dar\'oczy, 1963) and modify them so that entropy retains its natural meaning in the presence of negative measure.  We obtain a characterization of what we called signed R\'enyi $\alpha$-entropy:
\begin{equation} \label{eq1}
H_\alpha(P) = - \frac{1}{\alpha - 1} \log_2 \bigl[ \frac{\sum_{i=1}^n |p_i|^\alpha}{|\sum_{i=1}^n p_i|} \bigr],
\end{equation}
where $P = (p_1, \ldots, p_n)$ is a (real-valued) signed measure on phase space and $\alpha > 0$ with $\alpha \not= 1$ is a free parameter.  Clearly, Equation (1) reduces to the usual formula for R\'enyi $\alpha$-entropy when $P$ is an ordinary probability measure.

An immediate question is why we do not adopt the obvious signed analog to Shannon entropy:
\begin{equation} \label{eq2}
H_1(P) = - \sum_{i=1}^n |p_i| \log_2|p_i|.
\end{equation}
The answer is that this formula fails the basic physical requirement of extensivity: $H_1(P*Q) \not= H_1(P) + H_1(Q)$.  This requirement seems basic for both information and thermodynamic interpretations of entropy.  (Notice that, as must then be the case, $H_\alpha(P)$ does not yield signed Shannon entropy in the $\alpha \rightarrow 1$ limit.  In fact, it will typically diverge at $\alpha = 1$.)

We further justify our choice of entropy for signed phase space by showing that:
\begin{quote}
\textbf{Cancellation Witness} $H_\alpha(P)$ witnesses the presence of cancellation in a signed measure; \newline

\textbf{Schur-Concavity} $H_\alpha(P)$ is Schur-concave for all $\alpha > 1$; \newline

\textbf{Quantum H-Theorem} $H_\alpha(P)$ is non-decreasing, if $\alpha > 1$, under a de-phasing assumption on dynamics; \newline

\textbf{Unitarity Analog} $H_2(P)$ is preserved under discrete Moyal-bracket evolution, paralleling preservation of Hilbert-space entropies under unitary evolution.
\end{quote}
\vspace{0.1in}

Signed R\'enyi entropies already underpin an entropic characterization of the qubit via an uncertainty principle (Brandenburger, La Mura, and Zoble, 2022) and feature in recent analyses of quantum and superquantum correlations (e.g., Aw et al., 2023; Onggadinata, Kurzynski, and Kaszlikowski, 2023a, 2023b; Onggadinata et al., 2024).  The links between Hilbert-space R\'enyi entropy and thermodynamics, many of which are
surveyed in Goold et al., 2016), have been explored in detail by Baez (2022), Bakiev, Nakshidze, and Savchenko (2020), Jizba and Arimitsu (2004), Lu and Grover (2019), and St\'ephan, Misguich, and Pasquier (2010).  A recent R\'enyi-like measure of interest is the $\alpha$-logarithmic negativity (Wang and Wilde, 2020) defined on Hermitian operators such as partial transposes.  Ji et al.~(2024) connect relative R\'enyi entropy to the quantum state exclusion task.  Our quadratic case $H_2$ relates to the pioneering result of Manfredi and Feix (2000), who showed that the square of the Wigner function is invariant under Moyal evolution in the continuous phase-space setting.  By supplying an axiomatic foundation for R\'enyi entropy on phase space and establishing some preliminary properties, we hope to point the way to other uses of phase-space R\'enyi entropy in quantum information, foundations, and thermodynamics.

In Section 2, we review the classic R\'enyi (1961) axioms for ordinary measures.  In Section 3, we introduce our modification of the axiom set to encompass signed measures and state our main characterization theorem.  In Sections 4 and 5, we show how our signed R\'enyi entropy witnesses negativity and satisfies Schur concavity, and we give examples to show that the signed Shannon entropy candidate fails both properties.  Section 6 uses our formula to state and prove an abstract-level quantum H-theorem.  Section 7 establishes the baseline that our signed R\'enyi entropy is conserved under Moyal-bracket evolution on phase space, which is a necessary check that we have identified a well-behaved analog to von Neumann (or R\'enyi) entropy on Hilbert space.  Section 8 examines a renormalized version of our formula and looks briefly at interpreting the $\alpha$ parameter as an inverse temperature.  Section 9 concludes.

\section{Axioms for R\'enyi Entropy}
We follow the approach in R\'enyi (1961) by first axiomatizing entropy for the class of non-negative measures with total weight less than or equal to $1$ and then specializing to probability measures.  Fix a finite set of states $X = \{x_{1},\ldots,x_{n} \}$.  A generalized probability measure on $X$ is a tuple $P = (p_1, \ldots, p_n)$ where each $p_i \in \mathbb{R}_+$ and $\sum_{i=1}^n p_i \leq 1$.  The quantity $w(P) = \sum_{i=1}^n p_i$ is called the weight of $P$.  We assume $w(P) > 0$.

Given two generalized probability measures $P = (p_1, \ldots, p_m)$ and $Q = (q_1 ,\ldots, q_n)$, we denote by $P*Q$ the generalized probability measure which is the direct product:
\begin{equation} (p_1q_1,\ldots,p_1q_n,\ldots,p_mq_1,\ldots,p_mq_n).
\end{equation}
Also, we denote by $P\cup Q$ the generalized probability measure which is the (direct) sum:
\begin{equation}
(p_1, \ldots, p_m, q_1, \ldots, q_n),
\end{equation}
whenever it is well-defined, i.e., $\sum_ip_i + \sum_jq_j \leq 1$.  Finally, we write $(p)$ for the generalized probability measure consisting of the single real number $p > 0$.

The R\'enyi (1961) axioms are:

\begin{axiom}
\textbf{(Symmetry)} \,\, $H(P)$ is a symmetric function of the elements of $P$.
\end{axiom}

\begin{axiom}
\textbf{(Continuity)} \,\, $H((p))$ is continuous for all $0 < p \leq 1$.
\end{axiom}
    
\begin{axiom}
\textbf{(Calibration)} \,\, $H((\frac{1}{2})) = 1$.
\end{axiom}

\begin{axiom}
\textbf{(Extensivity)} \,\, $H(P * Q) = H(P) + H(Q)$.
\end{axiom}

\begin{axiom}
\textbf{(Mean-Value Property)} \,\, There is a strictly monotone and continuous function $g:\mathbb{R}\rightarrow\mathbb{R}$ such that for any 
$P,Q$, whenever $H(P\cup Q)$ is well-defined:
\begin{equation}
H(P\cup Q)=g^{-1}\big[\frac{w(P)g(H(P))+w(Q)g(H(Q))}{w(P\cup Q)}\big].
\end{equation}
\end{axiom}

R\'enyi (1961, Theorem 2) proved that if the function $g$ in Axiom 5 takes the form $g(x) = - dx + e$, for constants $d \not = 0$ and $e$, then Axioms 1-5 characterize the entropy functional:
\begin{equation}~\label{shannon}
H_1(P) = - \frac{\sum_{i=1}^n p_i \log_2 p_i}{\sum_{i=1}^n p_i},
\end{equation}
while, if $g(x) = - d2^{(1 -\alpha)x} + e$ for $\alpha > 0$ with $\alpha \not = 1$, then Axioms 1-5 characterize the family of entropy functionals:
\begin{equation}~\label{renyi}
H_\alpha(P) = - \frac{1}{\alpha - 1} \log_2 \big[\frac{\sum_{i=1}^n p_i^\alpha}{\sum_{i=1}^n p_i}\big].
\end{equation}
The entropies in Equations~\ref{shannon} and~\ref{renyi} are the familiar Shannon and R\'enyi entropies, defined for generalized probability measures.  As is well known, they can be related via an application of l'H\^opital's rule:
\begin{equation}~\label{hopital}
\lim_{\alpha \rightarrow 1} H_\alpha(Q) = H_1(Q).
\end{equation}

Dar\'oczy (1963) proved that Axioms 1-5 are, in fact, fully characterized by the two functional forms $g(x) = - dx + e$ and $g(x) = - d2^{(1 -\alpha)x} + e$ with the preceding parameter restrictions on $d, e$, and $\alpha$.  (He also observed that Axiom 1 is not needed, because it follows from Axiom 5.)

\section{Extension to Signed Measures}
A signed measure $P$ on $X$ is a tuple $P = (p_1, \ldots, p_n)$ where each $p_i \in \mathbb{R}$. Thus, the components $p_i$ are no longer required to be non-negative.  We impose $w(P) \neq 0$ but not $w(P) = 1$ (except when $P$ is a signed probability measure). The notations $P * Q$ and $P \cup Q$ have the same meanings as in the previous section and the latter is well-defined when $\sum_i p_i + \sum_j q_j \neq 0$. We now state our modified set of axioms.
\vspace{0.1in}

\noindent\textbf{Axiom 0. (Real-Valuedness)} \textit{$H(P)$ is a finite real number.}
\vspace{0.1in}

\noindent\textbf{Axiom 2$^\prime$. (Continuity)} \textit{$H((p))$ is continuous for all $p \neq 0$.}
\vspace{0.1in}

\noindent\textbf{Axiom 3. (Calibration)} \textit{$H\left(\left(\frac{1}{2}\right)\right) = 1$.}
\vspace{0.1in}

\noindent\textbf{Axiom 4. (Extensivity)} \textit{$H(P * Q) = H(P) + H(Q)$.}
\vspace{0.1in}

\noindent\textbf{Axiom 5$^\prime$. (Mean-Value Property)}  \textit{There is a continuous, strictly monotone function $g : \mathbb{R} \to \mathbb{R}$ such that any $P,Q$ whenever $w(P \cup Q)\neq 0$ is well-defined:}
\begin{equation} \label{eq:renyi}
H(P \cup Q) = g^{-1}\bigl[\frac{|w(P)| g(H(P)) + |w(Q)| g(H(Q))} {|w(P) + w(Q)|}
\bigr].
\end{equation}

Some comments on the axioms.  Axiom 0 ensures that entropy, as a measure of uncertainty or lack of information, is real-valued, which is critical for signed measures where standard entropy formulas may yield complex values due to negative components.  Axiom 2$^\prime$ extends Axiom 2 to signed measures, excluding $p = 0$.  Axioms 3 and 4 are as in R\'enyi (1961), stated for signed measures.  Axiom 5 deserves deserves some discussion.  It reduces to the R\'enyi mean-value axiom when $P, Q$ are non-negative measures.  In this case, the coefficients in the mean-value formula are the subsystem weights $w(P)$ and $w(Q)$, which represent how much of each subsystem is being mixed.  In the signed case, the quantities $w(P)$, $w(Q)$ and $w(P \cup Q)$ are not directly weights.  For example, a negative value of $w(P)$ does not represent a ``negative number of subsystems," but, rather, a contribution that cancels against other contributions.  Thus, the physically meaningful measure of the size of a signed subsystem is $|w(P)|$, and similarly for $Q$, and these are the quantities that appear in the numerator of Equation \ref{eq:renyi}.  But cancellation are key physically when systems are composed as $P \cup Q$, which is why the denominator first forms the sum $w(P) + w(Q)$ and then takes the absolute value.  (If we replaced the denominator by $|w(P)| + |w(Q)|$, we would lose much of architecture of the paper -- in particular, our H-theorem, and invariance under discrete Moyal evolution would both break.)

\begin{theorem} \label{th1}
Axioms 0, 2$^\prime$, 3, 4, and 5$^\prime$ hold if and only if $H(P)$ is given by:
\begin{equation} \label{eq10}
H_\alpha(P) = -\frac{1}{\alpha - 1} \log_2 \bigl[ \frac{\sum_{i=1}^n |p_i|^\alpha}{\left| \sum_{i=1}^n p_i \right|} \bigr],
\end{equation}
where $\alpha \in \mathbb{R}$ is a free parameter with $\alpha > 0$ and $\alpha \neq 1$.  The case $\alpha = 0$ is also admissible under the convention $\sum_i |p_i|^0 = \#\{i : p_i \not= 0\}$ (so we obtain a signed Hartley entropy).
\end{theorem}

The proof of Theorem \ref{th1} is in the Appendix.  Here, we make some observations.  An obvious signed analog to ordinary Shannon entropy is:
\begin{equation} \label{eq11}
H_1(P) = -\frac{\sum_{i=1}^n |p_i| \log_2 |p_i|}{| \sum_{i=1}^n p_i|},
\end{equation}
defined when $w(P) \not= 0$.  However, this formula violates the Extensivity axiom.  (From now on, all logs unless otherwise noted will be  to base $2$.)

\begin{Ex} \label{ex1}
Let $P$ and $Q$ be the same signed probability measure $(2, -1)$.  Then:
\begin{equation}
H_1(P) = H_1(Q) = -(2 \log 2 + 1 \log 1) = -2.
\end{equation}
For the product measure $P*Q = (4, -2, -2, 1)$:
\begin{equation}
H_1(P*Q) = -(4\log 4 + 2 \cdot 2\log 2 + 1 \log1) = -12 \not= -4 = H_1(P) + H_1(Q),
\end{equation}
showing that Axiom 4 does not hold.  By contrast, for signed R\'enyi $\alpha$-entropy:
\begin{equation}
H_\alpha(P) = H_\alpha(Q) = - \frac{1}{\alpha - 1}\log(2^\alpha + 1),
\end{equation}
and:
\begin{equation}
H_\alpha(P*Q) = - \frac{1}{\alpha - 1} \log(4^\alpha + 2 \cdot 2^\alpha + 1) = - \frac{1}{\alpha - 1} \log\bigl[(2^\alpha + 1)^2\bigr],
\end{equation}
consistent with Axiom 4.
\end{Ex}

We see the failure of signed Shannon entropy to obey Extensivity as significant evidence against its value as an entropy measure -- and as strong evidence in favor of the value of signed R\'enyi entropy.  Nevertheless, we will hold signed Shannon entropy in mind as we explore various properties of signed R\'enyi entropy, so that we can compare their performance more broadly.

Note from Example \ref{ex1} that signed Shannon entropy (Equation \ref{eq11}) cannot arise as the $\alpha \rightarrow 1$ limit of signed R\'enyi entropy $H_\alpha(P)$.  Indeed, it can be seen from Equation \ref{eq10} that signed R\'enyi entropy will generally diverge at $\alpha = 1$.  We do not see this as a shortcoming of our formula, which explicitly rules out $\alpha = 1$.  (But see Section 8.)

A closely related paper to this one is Koukouledikis and Jennings (2022).  These authors do not take an axiomatic approach, but, instead, directly assume that $\alpha = 2a / (2b - 1)$, where $a, b$ are positive integers with $a \geq b$, so that $p_i^\alpha$ is non-negative real-valued.  Our entropy formula is defined for all $\alpha > 0$, but it coincides with the Koukouledikis-Jennings one when $\alpha$ is restricted as in their paper.  We also treat unnormalized measures.  In the next two sections, we parallel Koukouledikis and Jennings (2022) in examining how non-classicality of probabilities is witnessed and in examining Schur-concavity properties.

\section{Witnessing Cancellation}
In this section, we show that signed R\'enyi entropy as just defined witnesses the presence of cancellation in signed measures.  This result is related to one on witnessing negativity in Koukouledikis and Jennings (2020, Theorem 10).  As we will see, signed Shannon entropy does not witness cancellation -- at least, not under our approach.  Given a signed measure $P = (p_1, \dots, p_n)$, write the total variation as $A_1(P) = \sum_i |p_i|$, and define the negativity:
\begin{equation} \label{eq:negativity}
N(P) = A_1(P) - |w(P)|.
\end{equation}
We have $N(P) \ge 0$ by the triangle inequality.  Note that our cancellation witness is invariant under a global sign flip $P \to -P$.  Indeed, $N(P) = 0$ whenever all entries of $P$ have the same sign, even if every $p_i$ is strictly negative.  Thus $N(P) > 0$ detects not the presence of a negative component per se, but the presence of both positive and negative mass -- that is, genuine cancellation structure in the signed measure.  Our use of the term ``negativity" should be understood in this sense.

\begin{theorem} \label{th2}
The signed measure $P$ exhibits negativity, i.e., $N(P) > 0$, if and only if there is an $\varepsilon > 0$ such that:
\begin{equation} \label{eq:negativity}
H_\alpha(P) < -\log_2|w(P)|,
\end{equation}
for all $\alpha \in(1, 1 + \varepsilon]$.
\end{theorem}

\begin{proof}
For $\alpha > 1$, the inequality $H_\alpha(P) < - \log_2 |w|$ is equivalent to $S_\alpha > |w|^\alpha$, where $S_\alpha = \sum_i |p_i|^\alpha$.  Define $f(\alpha) = S_\alpha - |w|^\alpha$.  Since each term $|p_i|^\alpha$ and $|w|^\alpha$ is continuous in $\alpha$, the function $f$ extends continuously to $[1,\infty)$.  Moreover:
\begin{equation}
f(1) = S_1 - |w| = \sum_i |p_i| - |w| = N(P) > 0.
\end{equation}
By continuity, there is an $\varepsilon > 0$ such that $f(\alpha)>0$ for all $\alpha \in (1, 1 + \varepsilon]$, which is equivalent to $H_\alpha(P) < - \log_2 |w|$ on this interval.

Conversely, suppose there is an $\varepsilon > 0$ such that Inequality \ref{eq:negativity} holds for all $\alpha \in(1, 1 + \varepsilon]$.  As before, this inequality is equivalent to $S_\alpha > |w(P)|^\alpha$.  Letting $\alpha \downarrow 1$ and using continuity of $S_\alpha$ and $|w|^\alpha$ in $\alpha$:
\begin{equation}
S_1 = \sum_i |p_i| \ge |w|^1 = |w|.
\end{equation}
If $S_1 = |w|$, then all $p_i$ have the same sign, and we may write $s_i = |p_i| \ge 0$ with $\sum_i s_i = |w|$.  Define $\lambda_i = s_i/|w|$, so that $\sum_i \lambda_i = 1$.  For $\alpha > 1$, the function $x \mapsto x^\alpha$ is strictly convex, and therefore:
\begin{equation}
\sum_i s_i^\alpha = |w|^\alpha \, \sum_i \lambda_i^\alpha \le |w|^\alpha \, \sum_i \lambda_i = |w|^\alpha,
\end{equation}
with equality only if exactly one $\lambda_i = 1$ and all others vanish, that is, only if $P$ has a single nonzero component.  This contradicts $S_\alpha > |w|^\alpha$ for $\alpha > 1$.  Therefore, $S_1 > |w|$, and so $N(P) > 0$, as required.
\end{proof}

The next example shows that signed Shannon entropy does not witness cancellation the same way.

\begin{Ex}~\label{ex2}
Consider the signed probability measure $P = (2, - \frac{1}{n}, \ldots, - \frac{1}{n})$ where the term $- \frac{1}{n}$ is repeated $n$ times.  Its total weight is $w(P) = 2 - n\cdot \frac{1}{n} = 1$, while its total variation is $\sum_i |p_i| = 2 + 1 = 3$, so that $N(P) = 2 > 0$.  Thus $P$ exhibits cancellation in the sense of Theorem \ref{th2}.

Signed Shannon entropy gives:
\begin{equation}~\label{shannon-no-witness}
H_1(P) = - 2 \log2 - n \frac{1}{n} \log(\frac{1}{n}) = \log n - 2,
\end{equation}
which is strictly positive as long as $n > 4$, and therefore does not detect the presence of cancellation in the sense of Theorem \ref{th2}.

By contrast, for signed R\'enyi entropy we find, for small $\varepsilon > 0$:
\begin{equation}~\label{renyi-yes-witness}
H_{1 + \varepsilon}(P) = - \frac{1}{\varepsilon} \log \bigl[2^{1 + \varepsilon} + n \times (\frac{1}{n})^{1 + \varepsilon}\bigr] \approx - \frac{1}{\varepsilon} \log 3 < 0
\end{equation}
Since $- \log_2 |w(P)| = 0$, this verifies the inequality
\begin{equation}
H_{1 + \varepsilon}(P) < -\log_2 |w(P)|,
\end{equation}
required by Theorem \ref{th2}.  Hence signed R\'enyi entropy correctly witnesses cancellation in this example.
\end{Ex}

\section{Schur-Concavity}
We now turn to majorization, that is, the notion that one measure is ``more uniform" than another one.  This is a basic test for any candidate for an entropy measure.  Recall that a function $f : \cal D \rightarrow \mathbb{R}$, where ${\cal D} \subseteq \mathbb{R}^n$, is Schur-concave on $\cal D$ if $f(P) \geq f(Q)$ whenever $P \in \cal D$ is majorized by $Q \in \cal D$, that is:
\begin{align}
\sum_{i = 1}^m p_{[i]} &\leq \sum_{i = 1}^m q_{[i]} \,\, \text{for all} \,\, m \leq n - 1, ~\label{majorization-a} \\
\sum_{i = 1}^n p_{[i]} &= \sum_{i = 1}^n q_{[i]}, ~\label{majorization-b}
\end{align}
where $(p_{[1]}, p_{[2]}, \ldots, p_{[n]})$ is a decreasing rearrangement of $P$, and similarly for $Q$.

\begin{theorem}~\label{th3}
If $\alpha > 1$, then signed R\'enyi entropy $H_\alpha$ is Schur-concave on $\{P \in \mathbb{R}^n : \sum_i p_i \not= 0\}$.
\end{theorem}

\begin{proof}
We first consider $P, Q \in \mathbb{R}^n$ without imposing the restrictions $\sum_i p_i \not=0$ and $\sum_i q_i \not=0$.  Since $|p_i|^\alpha$ is convex for all $\alpha > 1$, it follows from Marshall, Olkin, and Arnold (2011, Proposition C.1, p.92) that $\sum_i |p_i|^\alpha$ is Schur-convex on $\mathbb{R}^n$.  Therefore, if $P$ is majorized by $Q$:
\begin{equation} ~\label{schur-proof}
\sum_{i=1}^n |p_i|^\alpha \leq \sum_{i=1}^n |q_i|^\alpha.
\end{equation}
Now assume $\sum_i p_i = \sum_i q_i = w \not= 0$.  The map $x \mapsto - 1/(\alpha - 1) \times \log(x/w)$ is strictly decreasing for $\alpha > 1$.  Applying it to both sides of Equation \ref{schur-proof} reverses the inequality, yielding $H_\alpha(P) \geq H_\alpha(Q)$.  This establishes that $H_\alpha$ is Schur-concave on $\{P \in \mathbb{R}^n : \sum_i p_i  = w \not= 0\}$.
\end{proof}

Koukouledikis and Jennings (2022, Theorem 9) state that R\'enyi entropy for the parametric family $\alpha = 2a/(2b - 1)$, where $a, b$ are positive integers with $a \geq b$, is Schur-concave for signed probability measures.  Our Theorem \ref{th3} extends the domain of their result.

\begin{Ex} \label{ex3}
Consider the signed probability measures $P = (0.08, 0.45, 0.47)$ and $Q =(-0.3, 0.6, 0.7)$.  Then $P$ is majorized by $Q$.  The corresponding signed Shannon entropies are:
\begin{align}
H_1(P) &= - (0.08 \log 0.08 + 0.45 \log 0.45 + 0.47 \log 0.47) \approx 1.3219, \\
H_1(Q) &= - (0.3 \log 0.3 + 0.6 \log 0.6 + 0.7 \log 0.7) \approx 1.3235,
\end{align}
so that $H_1(P) < H_1(Q)$, contradicting Schur-concavity.  By contrast, we find the signed R\'enyi $2$-entropies satisfy $H_2(P) \approx 1.2183$ and $H_2(Q) \approx 0.0893$, in line with Theorem \ref{th3}.
\end{Ex}

\section{A Quantum H-Theorem}
We now make use of our concept of signed R\'enyi entropy is to obtain a quantum H-theorem.  The puzzle of what, given unitarity, is the appropriate quantum analog to the classical Boltzmann H-theorem has a long history.  Pauli (1928) developed a time-irreversible quantum master equation by assuming that off-diagonal elements of the density matrix vanish over time owing to random-phase cancellation (Cabrera, 2017).  Von Neumann (1929) located entropy increase in the act of measurement, which introduces classical uncertainty.  Modern approaches mostly make use of entanglement between a system and an environment to bring about increasing entropy, but, to make the issue meaningful, the system should be ``quasi-isolated" (Lesovik et al., 2016).  More precisely, the time scale is long enough to allow de-phasing of off-diagonal elements of the density matrix, but short enough to ensure negligible energy exchange with the environment (Lesovik et al., op.cit.).

In this section, we formulate an abstract H-theorem for signed measures.  Our theorem can be thought of as encompassing, for example, the Wigner quasi-probability representation of quantum systems on phase space (Wigner, 1932) and other phase-space representations.  (See Rundle and Everitt, 2021 for a recent survey.)  More precisely, it will cover general evolution of any no-signaling system (Popescu and Rohrlich, 1994).  This follows from Abramsky and Brandenburger (2011, Theorem 5.9), who prove an equivalence between no-signaling systems and signed phase-space representations.  Also relevant here is van de Wetering (2018), who establishes quantum theory as a subcategory of the category of quasi-stochastic processes.

To proceed, fix again a finite set of states $X = \{x_1, \ldots, x_n\}$, to be thought of as the set of possible phase states of a physical system.  We allow for signed measures $P = (p_1, \ldots, p_n)$ on states in $X$.  We next introduce transition rates $\lambda_{ij}$ between states $x_i$ and $x_j$, for $i, j = 1, \ldots, n$.  These rates are defined as limits:
\begin{equation}~\label{transition-rate}
\lambda_{ij} = \lim_{\Delta t \rightarrow 0} \frac{{\rm Pr}(x_j \,\, \text {at time} \,\, t + \Delta t \, | \, x_i \,\, \text {at time} \,\, t) - \delta_{ij}}{\Delta t},
\end{equation}
where $\delta_{ij}$ is the Dirac delta function: $\delta_{ij} = 1$ if $i = j$ and $0$ otherwise.  As is customary, this limit is assumed to be well-defined (see, e.g., Whittle, 2000, p.183-184 for a discussion).  Observe that $\sum_{j=1}^n \lambda_{ij} = 0$ for each $i$, and, for a classical system, $\lambda_{ij} \geq 0$ for $i \not= j$.

To allow for fully general evolution of a signed measure on phase space, we first write down a master equation where $\lambda_{ij} < 0$ is permitted.  Call these signed transition rates.  We impose micro-reversibility (Whittle, 2000, p.199): $\lambda_{ij} = \lambda_{ji}$ for all $i, j$.  This assumption can be justified in various ways.  In quantum mechanics, symmetry of transition rates is a consequence of Hermiticity of the Hamiltonian (Schwartz, 2021, Section 3.4).  Our general master equation is:
\begin{align}
\frac{d p_i}{d t} &= \sum_{j=1}^n \lambda_{ji} p_j(t) \\
&= \sum_{j=1}^n \lambda_{ji} p_j(t) - \sum_{j=1}^n \lambda_{ij} p_i(t) \\
&= \sum_{j=1}^n \lambda_{ij} \bigl[ p_j(t) - p_i(t)\bigr], ~\label{master-equation-c}
\end{align}
for $i = 1, \ldots, n$.  This is standard except that $P$ is a signed measure, and the matrix $\Lambda$ of transition rates may have negative off-diagonal entries.

The goal now is to identify a (minimal) condition on our master equation that yields a Second Law for signed R\'enyi entropy.  The next theorem shows that such a condition is classicality of the evolution of the phase-space measure.  Of course, this does not say that that phase-space measures must be non-negative, only that their evolution is ruled by non-negative probabilities.  We state our result and then add some interpretation.

For our result, we consider a domain $\cal D$ in $\{ P \in \mathbb{R}^n : w(P) \not= 0\}$ that is closed under permutations, is convex, and has a non-empty interior.  This will enable us to appeal to a differential characterization of Schur-concavity. 

\begin{theorem}~\label{th4}
Assume the off-diagonal entries in the matrix $\Lambda$ of (signed) transition rates are non-negative -- that is, assume classical evolution.  Then signed R\'enyi $\alpha$-entropy, for $\alpha > 1$, is non-decreasing on $\cal D$, that is: $d H_\alpha/d t \geq 0$.
\end{theorem}

\begin{proof}
We find the time derivative of signed R\'enyi entropy:
\begin{equation}~\label{time-derivative-a}
\frac{d H_\alpha}{d t} = \sum_{i=1}^n \frac{\partial H_\alpha}{\partial p_i} \frac{d p_i}{d t}
= \sum_{i=1}^n \frac{\partial H_\alpha}{\partial p_i} \sum_{j=1}^n \lambda_{ij} \bigl[p_j(t) - p_i(t)\bigr],
\end{equation}
where we have substituted in Equation~\ref{master-equation-c}.  Interchanging $i$ and $j$, we can also write:
\begin{equation}~\label{time-derivative-b}
\frac{d H_\alpha}{d t} = \sum_{j=1}^n \frac{\partial H_\alpha}{\partial p_j} \frac{d p_j}{d t}
= \sum_{j=1}^n \frac{\partial H_\alpha}{\partial p_j} \sum_{i=1}^n \lambda_{ji} \bigl[p_i(t) - p_j(t)\bigr].
\end{equation}

Adding Equations~\ref{time-derivative-a} and~\ref{time-derivative-b} and using symmetry, we arrive at:
\begin{equation}~\label{time-evolution}
\frac{d H_\alpha}{d t} = - \frac{1}{2} \sum_{i=1}^n \sum_{j=1}^n \lambda_{ij} \bigl(\frac{\partial H_\alpha}{\partial p_i} - \frac{\partial H_\alpha}{\partial p_j}\bigr) \bigl[p_i(t) - p_j(t)\bigr].
\end{equation}

By Theorem \ref{th3}, signed R\'enyi entropy $H_\alpha$, for $\alpha > 1$, is Schur-concave.  Now, use the fact that a quotient of continuous functions with a strictly positive denominator is continuous, to see that $H_\alpha$ is continuous on $\cal D$.  Next calculate:
\begin{equation} \label{eqc1}
\frac{\partial H_\alpha}{\partial p_j} = - \frac{1}{(\alpha - 1)\log_e2} \Bigl[\frac{\alpha |p_j|^{\alpha - 1}\text{sign}(p_j)}{\sum_i |p_i|^\alpha} - \frac{\text{sign}(\sum_i p_i)}{|\sum_i p_i|} \Bigr].
\end{equation}
The key observation here is that $\alpha|p_j|^{\alpha - 1}\text{sign}(p_j)$ is continuous on $\mathbb{R}$ precisely when $\alpha > 1$, and the operations in Equation \ref{eqc1} preserve continuity.  We conclude that $H_\alpha$ is $C^1$ on the interior of $\cal D$.

Using these properties, we can appeal to Marshall, Olkin, and Arnold (2011, Theorem A.4.a, pp.84-85) to obtain that for all $i, j$:
\begin{equation}~\label{schur-inequality-b}
\bigl(\frac{\partial H_\alpha}{\partial p_i} - \frac{\partial H_\alpha}{\partial p_j}\bigr) \bigl(q_i - q_j\bigr) \leq 0,
\end{equation}
on $\cal D$.  Since $\lambda_{ij} \geq 0$ for all $i \not= j$, it follows that $d H_\alpha/d t \geq 0$, as required.
\end{proof}

The key to Theorem \ref{th4} is the non-negativity of the transition rates.  This can be thought of as a classicalization assumption, which, over time, has a smoothing effect on the system and thereby suppresses its ability to exhibit negative measure.  We view our classicalization assumption as an abstraction of specific decoherence mechanisms posited for quantum systems in order to arrive at a quantum H-theorem.  See the references at the beginning of this section, and also Gemmer, Michel, and Mahler (2009) and Gogolin and Eisert (2016) for comprehensive treatments.

\begin{Ex}
We show that while signed R\'enyi entropy is non-decreasing, per Theorem \ref{th4}, signed Shannon entropy may not be.  Let the matrix of transition rates be given by:
\begin{equation}
\Lambda = \begin{pmatrix}
    -1 & 1/2 & 1/2 \\
    1/2 & -1 & 1/2 \\
    1/2 & 1/2 & -1
\end{pmatrix},
\end{equation}
and choose the initial signed probability measure: $P(0) = (-1/7, 3/7, 5/7)$.  Figure 1 shows the evolution of signed R\'enyi $2$-entropy $H_2(t)$ and signed Shannon entropy $H_1(t)$.  While signed Shannon entropy converges to the same limit $\log 3 \approx 1.585$ as signed R\'enyi $2$-entropy, we see that it is transiently non-monotonic.

\vspace{0.15in}
\hspace{60bp}\includegraphics[scale=0.25]{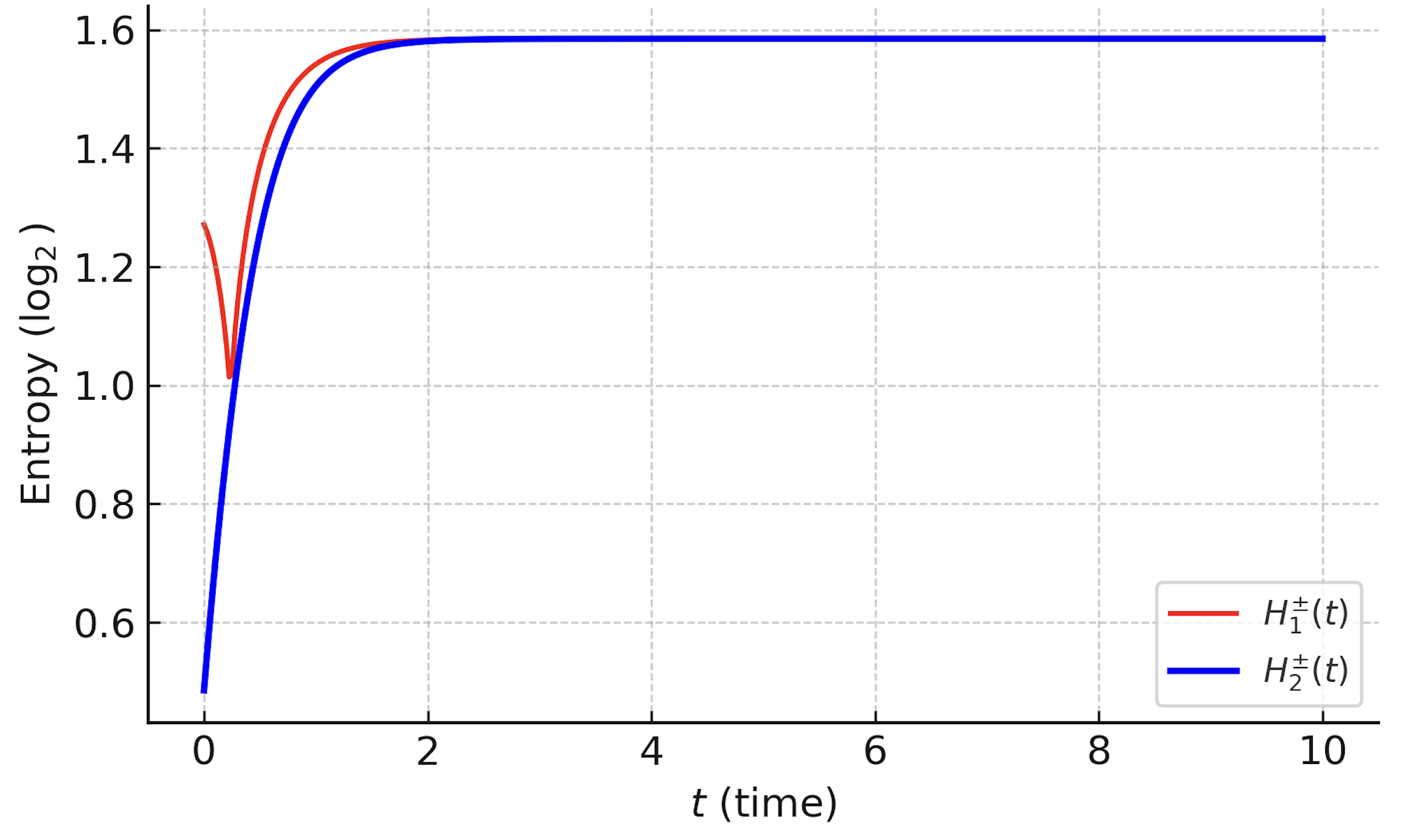}

\hspace{70bp}Figure 1: Evolution of Signed R\'enyi and Signed Shannon Entropies
\end{Ex}

\section{Discrete Moyal-Bracket Evolution}
Our H-Theorem hinged on non-negativity of off-diagonal entries of the matrix of transition rates -- an abstraction of decoherence mechanisms investigated in the literature.  An important check on this result is that signed R\'enyi entropy, for a suitably chosen free parameter, is constant and not increasing under discrete Moyal-bracket evolution on phase space.  This will constitute the phase-space analog to the the standard fact that von Neumann entropy $H_1(\rho) = - \text{Tr} (\rho \log \rho)$ -- and every R\'enyi entropy $H_\alpha(\rho) = -1/(\alpha - 1) \times \log \text{Tr}(\rho^\alpha)$ -- is constant under unitary evolution on Hilbert space.  The check will further underpin signed R\'enyi $\alpha$-entropy as a natural entropy for quasi-probability descriptions.

Fix the Hilbert-space dimension to be an odd-prime power and let $W$ denote the discrete Wigner function constructed from Wootters phase-point operators (Wootters, 1987) or, equivalently, from the finite field construction of Gibbons, Hoffman, and Wootters (2004).  (The argument can be extended to general finite phase spaces using the methods in Ferrie, 2011.)  Time evolution under a Hamiltonian $\cal H$ in this phase-space picture is governed by the discrete Moyal bracket (Gross, 2006):
\begin{equation}
\dot W = \{{\cal H}, W\}_{*d}.
\end{equation}
This evolution is implemented by a real skew-symmetric Liouvillian $L_{\cal H}$ acting on $W$:
\begin{equation} \label{eqmotion}
\dot W = L_{\cal H} W, \,\, L^{\mathsf{T}}_{\cal H} = - L_{\cal H}.
\end{equation}
Skew-symmetry holds because the discrete Moyal bracket inherits the antisymmetry of its continuous counterpart (Moyal, 1949).

\begin{theorem} \label{th5}
Fix a discrete Wigner function $W$ and a Hamiltonian $\cal H$.  Then $\sum_\xi W_\xi^2$ is a constant of motion, so that signed R\'enyi $2$-entropy is conserved, that is: $dH_2/dt = 0$.
\end{theorem}

\begin{proof}
Define the quadratic moment:
\begin{equation}
V(t) = \sum_\xi W_\xi(t)^2 = W(t)^{\mathsf{T}} \, W(t).
\end{equation}
Taking the time derivative:
\begin{equation}
\frac{dV}{dt} = \bigl(\frac{dW^{\mathsf{T}}}{dt}\bigr) \, W \, + \,  W^{\mathsf{T}} \bigl(\frac{dW}{dt}\bigr) = W^{\mathsf{T}} L_{\cal H}^{\mathsf{T}} \, W \, + \, W^{\mathsf{T}} L_{\cal H} W,
\end{equation}
where we interchange transpose and derivative and substitute in the first part of Equation \ref{eqmotion}.  Substituting in from the second part:
\begin{equation}
\frac{dV}{dt} = W^{\mathsf{T}} (-L_{\cal H}) \, W \, + \, W^{\mathsf{T}} L_{\cal H} W = 0,
\end{equation}
so that $V(t)$ is a constant of motion.  It follows that:
\begin{equation}
\frac{d}{dt} H_2\bigl(W(t)\bigr) = - \frac{d}{dt} \log V(t) = 0,
\end{equation}
as required.
\end{proof}

\section{Renormalization and Inverse Temperature}
In this section we note a renormalization formula for the case $\alpha = 1$ and also comment briefly on interpreting $\alpha$ as a (dimensionless) inverse temperature.  First, we can easily renormalize our signed R\'enyi entropy formula $H_\alpha(P)$ at $\alpha = 1$ by multiplying through by $\alpha - 1$ to obtain:
\begin{equation}
H_1(P)_{\text{Ren}} = - \log \bigl[ \frac{\sum_{i=1}^n |p_i|}{|\sum_{i=1}^n p_i|} \bigr].
\end{equation}
This formula satisfies Extensivity, witnesses negativity ($\sum_i|p_i| > |\sum_i p_i|$ if and only if at least one $p_i$ is negative), and is Schur-concave.  However, it will generally fail to be conserved under Moyal-bracket evolution and so cannot serve as a fundamental entropy in a theory whose elementary reversible dynamics are unitary.

Next, we interpret the signed R\'enyi index $\alpha$ as a dimensionless inverse temperature parameter by writing $\alpha = T_0/T$, where we set the Boltzmann constant $k_B = 1$ and $T_0$ is a fixed reference temperature.

\begin{Ex}
For ordinary probability measures, R\'enyi entropy is well known to be non-increasing in $\alpha$.  To examine the signed case, consider the signed probability measure $P = (\frac{4}{7}, - \frac{1}{7}, \frac{3}{14}, \frac{5}{14})$.  Figure 2 shows the behavior of signed R\'enyi entropy $H_\alpha(P)$ as a function of $1/\alpha$ over a range that produces nonmonotonicity.  (We chose a range that does not include the pole at $\alpha = 1$, since this is just a ``reference-point nonmonotonicity" independent of any system structure.)

\vspace{0.15in}
\hspace{60bp}\includegraphics[scale=0.55]{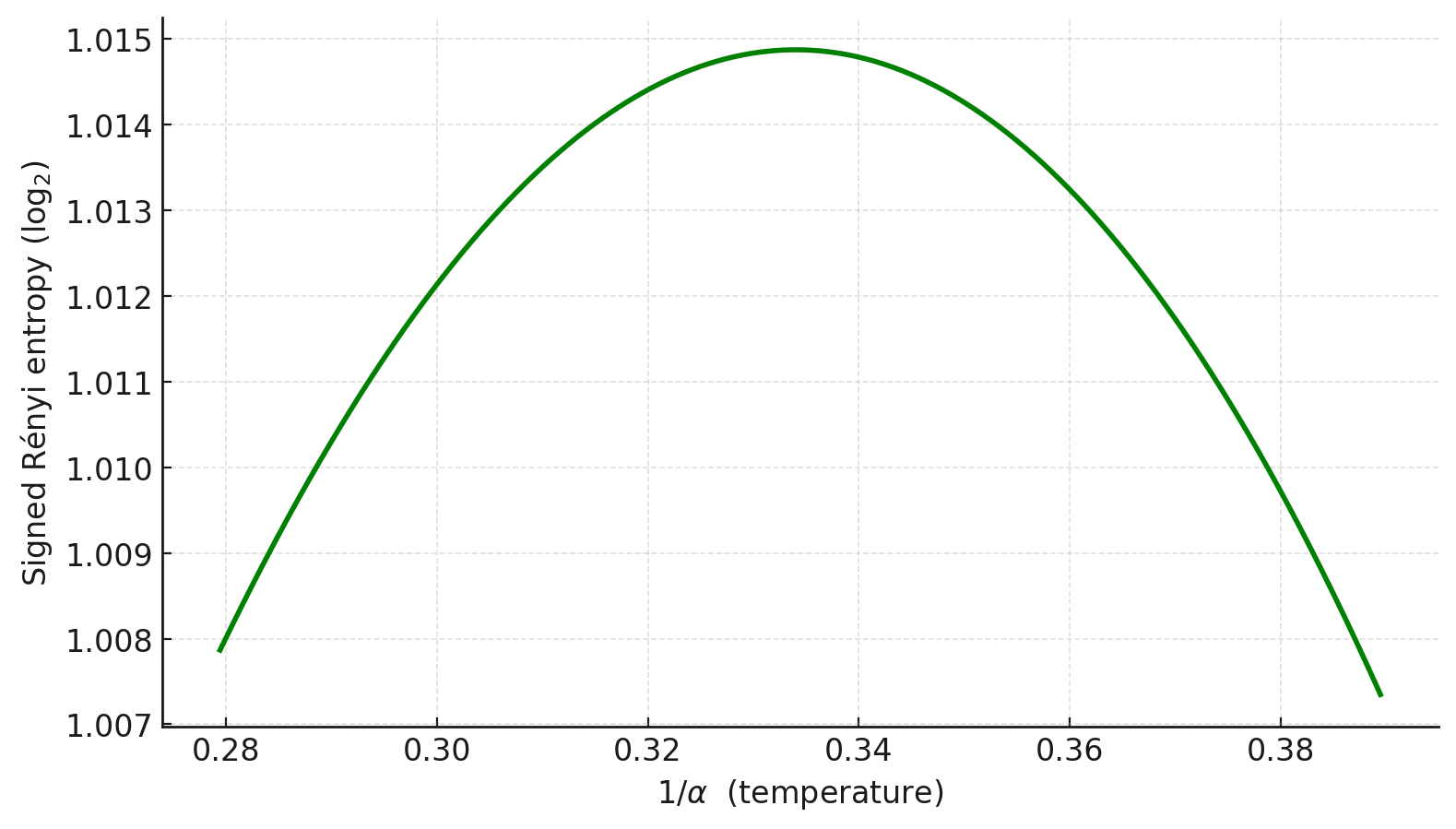}

\hspace{110bp}Figure 2: Nonmonotonicity of Signed R\'enyi Entropy
\end{Ex}

Interpreting $1/\alpha$ as temperature, we see that signed R\'enyi entropy first increases with temperature over the range considered, reaches a peak, and then decreases.  While this example is just a toy model, it can be connected conceptually to quantum systems that display intervals of negative specific heat (Defenu et al., 2023).  To the left of the peak, we are in the normal regime where increasing temperature goes with increasing entropy.  Beyond the peak, we are in an ``anomalous" negative specific-heat regime, where further heating decreases the entropy.  (We note, though, that for a fixed signed probability measure, $H_\alpha(P)$ can display at most one extremum on either side of the pole.  Obtaining an S-shaped caloric curve requires that the measure itself evolve; see Defenu et al., op.cit..)

\section{Conclusion}
We have provided a characterization of entropy for signed measures, arriving at the unique one-parameter family:
\begin{equation}
H_\alpha(P) = - \frac{1}{\alpha - 1} \log \bigl[ \frac{\sum_{i=1}^n |p_i|^\alpha}{|\sum_{i=1}^n p_i|} \bigr], \quad \alpha \ge 0 \,\, \text{and} \,\,  \alpha \not=1,
\end{equation}
singled out by our five information-theoretically interpretable axioms.  The ``obvious" signed Shannon entropy is excluded.

The operational virtues of our formula include: witnessing cancellation, which is closely related to the key quantum resource of negativity (Veitch et al., 2012; Chabaud, Emeriau, and Grosshans, 2021); Schur-concavity, i.e., monotonicity under majorization for all $\alpha > 1$, placing $H_\alpha$ in the family of entropy-like quantities that cannot decrease under mixing; again for $\alpha > 1$, a quantum H-theorem for classicalizing dynamics; and, for $\alpha = 2$, conservation under discrete Moyal-bracket (unitary) evolution.  We also commented briefly on the physical meaning of the fact that signed R\'enyi entropy is non-monotonic in the $\alpha$ parameter.

An open foundational question is the extension of our axiomatic approach to obtain a useful formula for infinite phase space, complementing existing continuous entropy measures (Albarelli et al., 2018).

\renewcommand{\thesection}{A}
\setcounter{equation}{0}
\renewcommand\theequation{A.\arabic{equation}}
\setcounter{theorem}{0}
\renewcommand\thetheorem{A.\arabic{theorem}}

\section*{Appendix}
In the section we prove the sufficiency direction of Theorem \ref{th1}.  The converse is checked by direct substitution.

\begin{lemma} \label{lem:cauchy}
Under Axioms 0, 2$^\prime$, 3, and 4, if $p \not= 0$, then $H((p)) = - \log|p|$.
\end{lemma}

\begin{proof}
Write $h(p) = H((p))$.  Axioms 0 and 2 imply that $h$ is real-valued and continuous except at $0$.  Axiom 4 implies that $h(pq) = h(p) + h(q)$ whenever $p, q \not= 0$.  This is a version of Cauchy's logarithmic functional equation, with general solution $h(p) = c \log|p|$, where $c$ is a real constant.  See Acz\'el and Dhombres (1989, Theorem 3, p.27).  Axiom 3 fixes $c = -1$.
\end{proof}

\begin{lemma} \label{induction}
Under Axiom 5$^\prime$, if $w(P) \not= 0$, then:
\begin{equation} \label{eq:mean-value-a}
H(P) = g^{-1}\Bigl(\frac{\sum_{i=1}^n |p_i| g(H((p_i)))}{|\sum_{i=1}^n p_i|}\Bigr).
\end{equation}
\end{lemma}

\begin{proof}
By induction on $n$.  For $n = 1$, $P = (p_1)$, $w(P) = p_1 \not= 0$, and:
\begin{equation}
H((p_1)) = g^{-1}\Bigl(\frac{|p_1| g(H((p_1)))}{|p_1|}\Bigr) = g^{-1}\bigl(g(H(((p_1)))\bigr),
\end{equation}
and so Equation \ref{eq:mean-value-a} holds.  Now assume inductively that Equation \ref{eq:mean-value-a} holds for measures with $(n-1)$ coordinates.  Let $ P= (p_1, \ldots, p_{n-1}, p_n)$ with $w(P) \not= 0$, and write $Q = (p_1, \ldots, p_{n-1})$, $R = (p_n)$, so that $P = Q \cup R$ and $w(P) = w(Q) + w(R) \not= 0$.  By the induction hypothesis:
\begin{equation}
H(Q) = g^{-1}\Bigl(\frac{\sum_{i=1}^{n-1} |p_i| g(H((p_i)))}{|w(Q)|}\Bigr).
\end{equation}
Applying Axiom 5$^\prime$ to $Q$ and $R$:
\begin{equation} \label{eq:mean-value-b}
H(P) = H(Q \cup R) = g^{-1}\Bigl(\frac{|w(Q)|g(H(Q)) + |w(R)|g(H(R))}{|w(Q) + w(R)|}\bigr).
\end{equation}
Substituting for $g(H(Q))$ and using $w(R) = p_n$, $H(R) = H((p_n))$, and $w(Q) + w(R) = \sum_{i=1}^n p_i$, we obtain:
\begin{align}
H(P) &= g^{-1}\Bigl(\frac{|w(Q)| \frac{\sum_{i=1}^{n-1} |p_i| g(H((p_i)))}{|w(Q)|} + |p_n| g(H((p_n)))}{|\sum_{i=1}^{n-1} p_i + p_n|}\Bigr) \\
&= g^{-1}\Bigl(\frac{\sum_{i=1}^{n-1} |p_i| g(H((p_i))) + |p_n| g(H((p_n)))}{|\sum_{i=1}^n p_i|}\Bigr),
\end{align}
which is Equation \ref{eq:mean-value-a} applied to $n$, as required.
\end{proof}

The next step is to restrict attention to non-negative measures $P$ (again, with $w(P) \not= 0$).  Our axioms reduce to the original R\'enyi (1961) axioms in this case, and so we can carry over the arguments in R\'enyi (1961) and Dar\'oczy (1963) to prove that $g : \mathbb{R} \to \mathbb{R}$ takes one of the forms:
\begin{align}
g(x) &=  - d x + e, \label{eq:affine} \\
g(x) &= d 2^{(1 -\alpha)x} + e, \label{eq:exponential}
\end{align}
where $d \not= 0$, $e$, and $\alpha \not= 1$ are arbitrary constants.

Since for every real number $x$, there is a singleton $(p)$ with $H((p)) = x$, the range of entropy values on non-negative measures is all of $\mathbb{R}$.  Hence the affine or exponential form of $g$ found by R\'enyi-Dar\'oczy holds on the entire domain of $g$.

\begin{lemma}
Under Axiom 4, the function $g$ cannot be affine.
\end{lemma}

\begin{proof}
If $g$ is affine as in Equation \ref{eq:affine}, then using Equation \ref{eq:mean-value-a} and Lemma \ref{lem:cauchy}, we get:
\begin{equation} \label{eq:finish-a}
- d \cdot H(P) + e = d \cdot \frac{\sum_i |p_i| \log |p_i|}{|\sum_i p_i|} + e \cdot \frac{\sum_i |p_i|}{|\sum_i p_i|},
\end{equation}
from which:
\begin{equation}
H(P) = - \frac{\sum_i |p_i| \log |p_i|}{|\sum_i p_i|} -\frac{e}{d}\Bigl(\frac{\sum_i |p_i|}{|\sum_i p_i|} - 1\Bigr).
\end{equation}
Now let $P = (\frac{1}{2}, \frac{1}{2})$ and $Q = (2, -1)$, so that $P*Q = (1, -\frac{1}{2}, 1, -\frac{1}{2})$.  We calculate:
\begin{equation}
H(P) = 1, \,\, H(Q) = -2 -\frac{2e}{d}, \,\, H(P*Q) = 1 - \frac{2e}{d},
\end{equation}
so that:
\begin{equation}
H(P*Q) \not= H(P) + H(Q),
\end{equation}
contradicting Axiom 4.
\end{proof}

\begin{lemma}
Under Axiom 4, the function $g$ must be exponential with $e = 0$.
\end{lemma}

\begin{proof}
If $g$ is exponential as in Equation \ref{eq:exponential}, then again using \ref{eq:mean-value-a} and Lemma \ref{lem:cauchy}, we get:
\begin{equation} \label{finish-b}
d \cdot 2^{(1 - \alpha)H(P)} + e = d \cdot \frac{\sum_i |p_i|^\alpha}{|\sum_i p_i|} + e \cdot \frac{\sum_i |p_i|}{|\sum_i p_i|},
\end{equation}
from which:
\begin{equation}
H(P) = - \frac{1}{\alpha - 1} \log \Bigl[ \frac{\sum_{i=1}^n |p_i|^\alpha}{\left| \sum_{i=1}^n p_i \right|} + \frac{e}{d}\bigl(\frac{\sum_i |p_i|}{|\sum_i p_i|} - 1\bigr)\Bigr].
\end{equation}
Again let $P = (\frac{1}{2}, \frac{1}{2})$ and $Q = (2, -1)$, so that $P*Q = (1, -\frac{1}{2}, 1, -\frac{1}{2})$.  We calculate:
\begin{equation}
H(P) = \log 2, \,\, H(Q) = - \frac{1}{\alpha - 1} \log\bigl[2^\alpha + 1 + \frac{2e}{d}\bigr], \,\, H(P*Q) = - \frac{1}{\alpha - 1} \log\bigl[2 + 2^{1 - \alpha} + \frac{2e}{d}\bigr].
\end{equation}
By Axiom 4:
\begin{equation}
2^{1 - \alpha} \times \bigl[2^\alpha + 1 + \frac{2e}{d}\bigr] = \bigl[2 + 2^{1 - \alpha} + \frac{2e}{d}\bigr],
\end{equation}
which, for any $\alpha \not= 1$, implies $e/d = 0$.
\end{proof}

\begin{lemma} \label{lem:alpha-nonneg}
Under Axiom 0, we must have $\alpha \ge 0$.  The case $\alpha = 0$ is admissible as a limit.
\end{lemma}

\begin{proof}
Suppose $\alpha < 0$. Consider the measure $(0, 1)$, which has
$w((0, 1)) = 1 \not= 0$. Then:
\begin{equation}
\sum_i |p_i|^\alpha = 0^\alpha + 1^\alpha = +\infty,
\end{equation}
so the argument of the log in Equation \ref{eq10} is infinite and $H((0, 1))$ is not finite.  This contradicts Axiom 0.  Thus $\alpha < 0$ is impossible.

For $\alpha = 0$, taking the limit $\alpha \to 0$ in Equation \ref{eq10} yields:
\begin{equation}
H_0(P) = \log_2 \bigl[\frac{\#\{i : p_i \not= 0\}}{|\sum_i p_i|}\bigr],
\end{equation}
which is finite.  One can check directly that $H_0$ satisfies Axioms 0--5$^\prime$, so $\alpha = 0$ is admitted as a limiting case.
\end{proof}

\section*{References}

\noindent Abramsky, S., and A. Brandenburger, ``The Sheaf-Theoretic Structure of Non-Locality and Contextuality," \textit{New Journal of Physics}, 13, 2011, 113036.
\vspace{0.1in}

\noindent Acz\'el, J., and J. Dhombres, \textit{Functional Equations in Several Variables}, Cambridge University Press, 1989, 26-27.
\vspace{0.1in}

\noindent Albarelli, F., M. Genoni, M. Paris, and A. Ferraro,
``Resource Theory of Quantum Non-Gaussianity and Wigner Negativity," \textit{Physical Review A}, 98, 2018, 052350.
\vspace{0.1in}

\noindent Aw, C., K. Onggadinata, D. Kaszlikowski, and V. Scarani, ``Quantum Bayesian Inference in Quasiprobability Representations," \textit{PRX Quantum}, 4, 2023, 020352.
\vspace{0.1in}

\noindent Baez, J., ``R\'enyi Entropy and Free Energy," \textit{Entropy}, 24, 2022, 706.
\vspace{0.1in}

\noindent Bakiev, T., D. Nakashidze, and A. Savchenko, ``Certain Relations in Statistical Physics Based on R\'enyi Entropy" \textit{Moscow University Physics Bulletin}, 75, 2020, 559-569.
\vspace{0.1in}

\noindent Brandenburger, A., P. La Mura, and S. Zoble, ``R\'enyi Entropy, Signed Probabilities, and the Qubit," \textit{Entropy}, 24, 2022, 1412.
\vspace{0.1in}

\noindent Cabrera, G., \textit{A Unifying Approach to Quantum Statistical Mechanics: A Quantum Description of Macroscopic Irreversibility}, 2017, at https://www.ifi.unicamp.br/$\sim$cabrera/.
\vspace{0.1in}

\noindent Chabaud, U., P.-E. Emeriau, and F. Grosshans, ``Witnessing Wigner Negativity," \textit{Quantum}, 5, 2021, 471.
\vspace{0.1in}

\noindent Coles, P., M. Berta, M. Tomamichel, and S. Wehner, ``Entropic Uncertainty Relations and Their Applications," \textit{Review of Modern Physics}, 89, 2017, 015002.
\vspace{0.1in}

\noindent Csisz\'ar, I., ``Axiomatic Characterizations of Information Measures," \textit{Entropy}, 10, 2008, 261-273.
\vspace{0.1in}

\noindent Dar\'oczy, Z., ``\"Uber die gemeinsame Charakterisierung der zu den nicht vollst\"andig en Verteilungen geh\"origen Entropien von Shannon und von R\'enyi," \textit{Zeitschrift f\"ur Wahrscheinlichkeitstheorie und verwandte Gebiete}, 1, 1963, 381-388.
\vspace{0.1in}

\noindent Defenu, N., T. Donner, T. Macr\'i, G. Pagano, S. Ruffo, and A. Trombettoni, ``Long-Range Interacting Systems," \textit{Reviews of Modern Physics}, 95, 2023, 035002.
\vspace{0.1in}

\noindent Ferrie, C., ``Quasi-Probability Representations of Quantum Theory with Applications to Quantum Information Science," \textit{Reports on Progress in Physics}, 74, 2011, 116001.
\vspace{0.1in}

\noindent Feynman, R., ``Negative Probability," in Hiley, B., and F. Peat (eds.), \textit{Quantum Implications: Essays in Honour of David Bohm}, Routledge \& Kegan Paul, 1987, pp.235-248.
\vspace{0.1in}

\noindent Gemmer, J. M. Michel, and G. Mahler, \textit{Quantum Thermodynamics}, Springer, 2nd edition, 2009.
\vspace{0.1in}

\noindent Gherardini, S., and G. De Chiara, ``Quasiprobabilities in Quantum Thermodynamics and Many-Body Systems," \textit{PRX Quantum}, 5, 2024, 030201.
\vspace{0.1in}

\noindent Gibbons, K., M. Hoffman, and W. Wootters, ``Discrete Phase Space Based on Finite Fields," \textit{Physical Review A}, 70, 2004, 62101.
\vspace{0.1in}

\noindent Gogolin, C., and J. Eisert, ``Equilibration, Thermalisation, and the Emergence of Statistical Mechanics in Closed Quantum Systems," \textit{Reports on Progress in Physics}, 79, 2016, 056001.
\vspace{0.1in}

\noindent Goold, J. M. Huber, A. Riera, L. del Rio, and P. Skrzypczyk, "The Role of Quantum Information in Thermodynamics -- A Topical Review," \textit{Journal of Physics A}, 49, 2016, 143001.
\vspace{0.1in}

\noindent Gross, D., ``Hudson's Theorem for Finite-Dimensional Quantum Systems," \textit{Journal of Mathematical Physics}, 47, 2006, 122107.
\vspace{0.1in}

\noindent Hardy, G., J. Littlewood, and G. P\'olya, \textit{Inequalities}, 2nd edition, Cambridge University Press, 1952.
\vspace{0.1in}

\noindent Ji, K., H. Mishra, M. Mosonyi, and M. Wilde, ``Barycentric Bounds on the Error Exponents of Quantum Hypothesis Exclusion, 2024, at https://arxiv.org/pdf/2407.13728.
\vspace{0.1in}

\noindent Jizba, P., and T. Arimitsu, ``The World According to R\'enyi: Thermodynamics of Multifractal Systems," \textit{Annals of Physics}, 312, 2004, 17-59.
\vspace{0.1in}

\noindent Koukoulekidis, N., and D. Jennings, ``Constraints on Magic State Protocols from the Statistical Mechanics of Wigner Negativity," \textit{npj Quantum Information}, 8, 2022, 42.
\vspace{0.1in}

\noindent Lesovik, G., A. Lebedev, I. Sadovskyy, M. Suslov, and V. Vinokur, ``H-Theorem in Quantum Physics," \textit{Scientific Reports}, 6, 2016, 32815.
\vspace{0.1in}

\noindent Lu, T.-C., and T. Grover, ``R\'enyi Entropy of Chaotic Eigenstates," \textit{Physical Review E}, 99, 2019, 032111.
\vspace{0.1in}

\noindent Marshall, A., I. Olkin, and B. Arnold, \textit{Inequalities: Theory of Majorization and Its Applications}, 2nd edition, Springer, 2011.
\vspace{0.1in}

\noindent Moyal, J., ``Quantum Mechanics as a Statistical Theory," \textit{Mathematical Proceedings of the Cambridge Philosophical Society}, 45, 1949, 99-124.
\vspace{0.1in}

\noindent Onggadinata, K., P. Kurzynski, and D. Kaszlikowski, ``Qubits from the Classical Collision Entropy," \textit{Physical Review A}, 107, 2023a, 032214.
\vspace{0.1in}

\noindent Onggadinata, K., P. Kurzynski, and D. Kaszlikowski, ``Simulations of Quantum Nonlocality with Local Negative Bits," \textit{Physical Review A}, 108, 2023b, 032204.
\vspace{0.1in}

\noindent Onggadinata, K., A. Tanggara, M. Gu, and D. Kaszlikowski, ``Negativity as a Resource for Memory Reduction in Stochastic Process Modeling," 2024, at https://arxiv.org/abs/2406.17292.
\vspace{0.1in}

\noindent Pauli, W., ``\"Uber das H-Theorem vom Anwachsen der Entropie vom Standpunkt der neuen Quantenmechanik," in Debye, P. (ed.), \textit{Probleme der modernen Physik: Arnold Sommerfeld zum 60, Geburtstag gewidmet}, Hirzel, 1928, 30-45.
\vspace{0.1in}

\noindent Popescu, S., and D. Rohrlich, ``Quantum Nonlocality as an Axiom," \textit{Foundations of Physics}, 24, 1994, 379-385.
\vspace{0.1in}

\noindent R\'enyi, A., ``On Measures of Information and Entropy," in Neyman, J. (ed.), \textit{Proceedings of the 4th Berkeley Symposium on Mathematical Statistics and Probability}, University of California Press, 1961, 547-561.
\vspace{0.1in}

\noindent R\'enyi, A., \textit{Probability Theory}, 1970, North-Holland and Akada\'emiai Kiad\'o.
\vspace{0.1in}

\noindent Rundle, R., and M. Everitt, ``Overview of the Phase Space Formulation of Quantum Mechanics with Application to Quantum Technologies," \textit{Advanced Quantum Technolgies}, 4, 2021, 2100016.
\vspace{0.1in}

\noindent Schwartz, M., ``Statistical Mechanics," Harvard University, Spring 2021, at https://scholar.harvard.edu/files/\newline schwartz/files/physics\_181\_lectures.pdf.
\vspace{0.1in}

\noindent Shannon, C., ``A Mathematical Theory of Communication," \textit{Bell System Technical Journal}, 27, 1948, 379-423 and 623-656.
\vspace{0.1in}

\noindent St\'ephan, J.-M., G. Misguich, and V. Pasquier, ``R\'enyi Entropy of a Line in Two-Dimensional Ising Models," \textit{Physical Review B}, 82, 2010, 125455.
\vspace{0.1in}

\noindent Veitch, V., C. Ferrie, D. Gross, and J. Emerson, ``Negative Quasi-Probability as a Resource for Quantum Computation," \textit{New Journal of Physics}, 14, 2012, 113011.
\vspace{0.1in}

\noindent van de Wetering, J., ``Quantum Theory is a Quasi-Stochastic Process Theory," in Coecke, B., and A. Kissinger (eds.), \textit{14th International Conference on Quantum Physics and Logic}, EPTCS 266, 2018, 179-196.\vspace{0.1in}

\noindent von Neumann, J., ``Beweis des Ergodensatzes und des H-Theorems in der neuen Mechanik," \textit{Zeitschrift f\"ur Physik}, 57, 1929, 30-70.
\vspace{0.1in}

\noindent von Neumann, J., \textit{Mathematische Grundlagen der Quantenmechanik}, Springer, 1932.
\vspace{0.1in}

\noindent Wang, X., and M. Wilde, ``$\alpha$-Logarithmic Negativity," \textit{Physical Review A}, 102, 2020, 032416.
\vspace{0.1in}

\noindent Whittle, P., \textit{Probability via Expectation}, Springer-Verlag, 4th edition, 2000.
\vspace{0.1in}

\noindent Wigner, E., ``On the Quantum Correction for Thermodynamic Equilibrium," \textit{Physical Review}, 40, 1932, 749.
\vspace{0.1in}

\noindent Wootters, W., ``A Wigner-Function Formulation of Finite-State Quantum Mechanics, \textit{Annals of Physics}, 176, 1987, 1-21.

\end{document}